\documentclass[envcountsame]{llncs}

\usepackage{microtype}
\usepackage{xspace,paralist}  
\usepackage{latexsym,amsmath,amssymb,amstext}

\spnewtheorem{requirement}[theorem]{Requirement}{\bfseries}{\itshape}

  {\begin{compactitem}[\decoone]\sl\small}%
  {\end{compactitem}}
\newcommand{\sfv}{\ensuremath{\mathbf{sfv}}}

\newcommand{\co}[1]{\hat{#1}}
\newcommand{\safe}[1]{#1^{\mathsf{S}}}
\newcommand{\key}[1]{\ensuremath{\mathop{\mathsf{#1}}\nolimits}\xspace}

\newcommand{\Case}{\key{case}}
\newcommand{\Of}{\key{of}}
\newcommand{\In}{\key{in}}
\newcommand{\Let}{\key{let}}
\newcommand{\Letstar}{\key{let*}}
\newcommand{\fold}{\key{fold}}
\newcommand{\unfold}{\key{unfold}}
\newcommand{\folds}[1]{\mathop{\safe{\mathsf{fold}}_{#1}}}
\newcommand{\unfolds}[1]{\mathop{\safe{\mathsf{unfold}}_{#1}}}

\newcommand{\ms}{\mathit{ms}}
\newcommand{\ns}{\mathit{ns}}

\newcommand{\type}[1]{\ensuremath{\mathbf{#1}}\xspace}
\newcommand{\Nat}{\type{nat}}

\newcommand{\Bit}{\type{bit}}
\newcommand{\Tree}{\type{tree}}
\newcommand{\Stream}{\type{stream}}
\newcommand{\Nats}{\type{nats}}

\newcommand{\Unit}{\ensuremath{\mathbf{unit}}\xspace}
\newcommand{\up}{\ensuremath{\mathord{\mathrm{up}}}\xspace}
\newcommand{\alt}{\mathrel{[\!]}}
\newcommand{\Ty}{\mathit{Ty}}

\newcommand{\constr}[1]{\ensuremath{\mathop{\mathsf{#1}}\nolimits}\xspace}
\newcommand{\Nought}{\constr{Nought}}
\newcommand{\One}{\constr{One}}
\newcommand{\Zero}{\constr{Zero}}
\newcommand{\Succ}{\constr{Succ}}
\newcommand{\Leaf}{\constr{Leaf}}
\newcommand{\Fork}{\constr{Fork}}

\newcommand{\declare}{\constr{declare}}
\newcommand{\data}{\constr{data}}
\newcommand{\codata}{\constr{codata}}
\newcommand{\Cons}{\constr{Cons}}

\newcommand{\id}{\mathrm{id}}

\newcommand{\asize}[1]{\ensuremath{\mathopen{|}#1\mathclose{|}}\xspace}
\newcommand{\osize}[1]{\ensuremath{\mathopen{\Vert}#1\mathclose{\Vert}}}

\newcommand{\lfp}[1]{\mu #1\,\sqdot\,}
\newcommand{\gfp}[1]{\nu #1\,\sqdot\,}
\newcommand{\Dcl}{\mathit{Dcl}}
\newcommand{\SetCat}{\ensuremath{\mathcal{S}\mkern-2.5mu\mathit{et}}\xspace}

\newcommand{\ustrut}{\rule[-.18\baselineskip]{0pt}{0pt}}

\newcommand{\inj}{\mathop{\iota}\nolimits}
\newcommand{\proj}{\mathop{\pi}\nolimits}
\newcommand{\upair}[2]{\underline{(}#1\underline{\ustrut,}#2\underline{)}}
\newcommand{\uiota}{\mathop{\underline{\iota}}\nolimits}
\newcommand{\uep}{\underline{()}}
\newcommand{\uc}{\underline{c}}

\newcommand{\RS}{\mathit{RS}}
\newcommand{\RSi}{\mathit{RS_1}}
\newcommand{\RSmi}{\mathit{RS^-_1}}

\newcommand{\Topic}[1]{\smallskip\noindent{\textbf{#1{.}}}\enspace}

\newcommand{\newfontobj}[2]{
  \newcommand{#1}[1]{
    \expandafter\def\csname##1\endcsname{{#2 ##1}}}}

\newcommand{\newthingie}[2]{
  \newcommand{#1}[1]{
    \expandafter\def\csname##1\endcsname{{#2 ##1}}}}

\newfontobj{\class}{\rm}
\newfontobj{\fnc}{\it}

\newfontobj{\category}{\bf}

\newcommand{\confined}{\mathclose{\upharpoonright}} 
\renewcommand{\gets}{\ensuremath{\mathrel{\colon=}}\xspace}

\newcommand{\dom}{\mathrm{dom}}

\newcommand{\image}{\mathrm{image}}
\newcommand{\set}[1]{\{\,#1\,\}}

\newcommand{\forallae}{\mathord{\stackrel{\kern 0.1em\infty}{\forall}}}

\newcommand{\entails}{\vdash}
\renewcommand{\phi}{\varphi}

\newcommand{\suchthat}{\mathrel{\,\stackrel{\rule{0.03em}{0.5ex}}%
{\rule[-.1ex]{0.03em}{0.5ex}}\,}}

\newcommand{\sqdot}{\rule{0.5mm}{0.5mm}}
\newcommand{\lam}[1]{\lambda #1\,\sqdot\,}

\newcommand{\fv}{\ensuremath{\mathbf{fv}}}

\newcommand{\is}{\ensuremath{\mathrel{\mathrel{:}\mathrel{:}=}}}
\newcommand{\synsep}{\;\;|\;\;}
\newcommand{\yields}{\mathbin{\downarrow}}

\newcommand{\of}{\colon}
\renewcommand{\colon}{\mathpunct{:}}
\newcommand{\Quad}[1]{\hspace*{#1em}}

\renewcommand{\today}{\number\day \space \ifcase\month\or
	January\or February\or March\or April\or May\or June\or
	July\or August\or September\or October\or November\or
	December\fi \space \number\year}

\newcommand{\irule}[2]{\ensuremath{{\frac{\strut\textstyle #1}%
  {\strut\textstyle #2}}}} 
\newcommand{\rulelabel}[1]{\hbox{\it #1:}\Quad{0.5}}
\newcommand{\sidecond}[1]{\Quad{0.5} \left(#1\right)}

\newcommand{\semantics}[1]{\ensuremath{[\![ #1 ]\!]}}

\newcommand{\level}{\textit{level}}

\makeatletter 
\def\rightmark{Danner and Royer: 
               Ramified Structural Recursion and Corecursion}
\def\leftmark{\rightmark}
\def\ps@myheadings{%
  \let\@mkboth\@gobbletwo
  \def\@oddhead{\hfill\small\rm\rightmark\qquad\thepage}%
  \def\@oddfoot{}
  \def\@evenhead{\small\thepage\qquad\rightmark\hfill}%
  \def\@evenfoot{}
  }
\ps@myheadings
\makeatother

\title{Ramified Structural Recursion and Corecursion \normalsize
 Extended Abstract}
\titlerunning{Ramified Structural Recursion and Corecursion}
\toctitle{Ramified Structural Recursion and Corecursion}
\author{Norman Danner \inst{1}
	    \and
       James S.~Royer \inst{2}
     }
          
\institute{%
  Department of Mathematics and Computer Science, 
  Wesleyan University, 
  Middletown, CT 06459, USA; 
  \email{ndanner@wesleyan.edu}
 \and
  Department of Electrical Engineering and Computer Science, 
  Syracuse University, 
  Syracuse, NY 13210, USA; 
  \email{jsroyer@syr.edu}
}

\begin{document}
\bibliographystyle{splncs03}

\maketitle

\begin{abstract} 
  We investigate feasible computation over a fairly general notion
  of data and codata. 
  Specifically, we present a direct Bellantoni-Cook-style
  normal/safe typed programming formalism, $\RSi$, that expresses
  feasible structural recursions and corecursions over data and
  codata specified by polynomial functors. (Lists, streams, finite
  trees, infinite trees, etc.~are all directly definable.)  A novel
  aspect of $\RSi$ is that it embraces structure-sharing as in
  standard functional-programming implementations.  As our data
  representations use sharing, our implementation of structural
  recursions are memoized to avoid the possibly exponentially-many
  repeated subcomputations a na\"{\i}ve implementation might
  perform.  We introduce notions of size for representations of data
  (accounting for sharing) and codata (using ideas from type-2
  computational complexity) and establish that type-level~1
  $\RSi$-functions have polynomial-bounded runtimes and satisfy a
  polynomial-time completeness condition.  Also, restricting 
  $\RSi$ terms to particular types produces characterizations of
  some standard complexity classes (e.g., $\omega$-regular
  languages, linear-space functions) and some less-standard classes
  (e.g., log-space streams).
\end{abstract}
  
\section{Introduction}

What counts as polynomial-time (much less ``feasible'') computation
over general forms of data is not a settled matter.  The
complexity-theoretic literature of higher-type computability is
still thin, it is  spotty on computation over codata (infinite
lists and trees) with some notable exceptions,\footnote{%
  Hartmanis and Stearns' paper \cite{HartmanisStearns65} that
  founded computational complexity largely focuses on the
  time-complexity of infinite streams as the authors directly
  adapted Turing's original machine model \cite{Turing36}
  which, recall, concerns stream-computation.}  
and even in the case of inductively defined data there are there
remain issues that are not that well explored (see the end of \S2
below).  We develop a notion of polynomial-time computation over
data and codata using a fairly simple implicit complexity formalism,
$\RSi$, that satisfies poly-time soundness and completeness
properties.  $\RSi$ is constructed in stages.  We first introduce 
$S^-$, a formalism for computing over inductively defined data by
classical structural (aka primitive) recursion.  $S^-$ has roughly 
the computational power of G\"{o}del's primitive recursive
functionals \cite{Longley:notions:1}.  To tame this power, we impose
a form of Bellantoni and Cook normal/safe ramification on $S^-$'s
structural recursions and obtain $\RSmi$, a system that satisfies
appropriate poly-time soundness and completeness properties.  We
next introduce $S$, an extension of $S^-$
to include codata definitions and
classical structural (aka primitive) corecursions.  We extend the
safe/normal ramification to corecursions and obtain $\RSi$ that
also satisfies  poly-time soundness and completeness
properties.  The subscript on $\RSmi$ and $\RSi$ is a reminder
that these formalisms focus on type-level 1 computation, eventhough
$\RSmi$ and $\RSi$ allow higher-type terms.  It turns out that by
restricting types in $\RSi$-terms, one can characterize other
complexity classes, e.g., $\omega$-regular languages, log-space
streams of characters, linear-space streams of strings, etc.  These
seem to be related to the \emph{two-sorted} complexity class
characterizations studied by Cook and Nguyen \cite{CookNguyen:book}.

\emph{Related Work.}
The Pola project of Burrell, Cockett, and Redmond
\cite{Burrell:2009,cockett:redmond:low:cat} has  aims similar to ours,
but Pola forbids any structure-sharing of safe-data or safe-codata.
$\RSmi$ and $\RSi$, in contrast, embrace structure-sharing and
adjust the implementation of structural recursions to accomodate it.
As a result $\RSi$ and Pola describe different notions of
polynomial-time over data and codata.  
How deep these differences go is an intriguing question.
Pola also has a well-developed categorical
semantics that,  at present, $\RSi$ notably lacks.
Ramyaa and Leivant \cite{Ramyaa:2010,Ramyaa:Leivant:2011} explore
feasible first-order stream programming formalisms.
In~\cite{Ramyaa:2010}, they use infinite binary trees with
string-labels to give a partial proof-theoretic characterization of
the type-2 basic feasible functionals ($\mathrm{BFF}_2$) of
Mehlhorn~\cite{Mehlhorn76} and Cook and
Urquhart~\cite{CookUrquhart:feasConstrArith}.  In
\cite{Ramyaa:Leivant:2011}, they give a definition of logspace
stream computation and a schema of ramified co-recurrence which
parallels Leivant's ramified recurrence of \cite{Leivant:FM2}, and
characterize logspace streams as those definable using $2$-tier
co-recurrences.  F\'{e}r\'{e}e \emph{et al.}~\cite{Feree:2010} also
consider stream computation, but primarily as a technical tool in
characterizing $\mathrm{BFF}_2$ as the functions computed by a
rewrite system over streams that has a second-order polynomial
interpretation.

%
	



\Topic{Background}
\emph{Pointer Machines.} 
We assume that the underlying model of computation is along the
lines of Kolmogorov and Uspenskii's 
``pointer machines'' or Sch\"{o}nhage's \emph{storage modification 
machines}
\cite{vanEmdeBoas90}.

\emph{Types.} 
The \emph{simple types} over a  set of base types $B$ are given by: 
$\Ty^B \is B$ $|$ $\Unit$ $|$ 
$\Ty^B+\Ty^B$ $|$ $\Ty^B\times\Ty^B$ $|$ $\Ty^B\to\Ty^B$, where $\Unit$ 
(which counts as a base type) is the type of the 
empty product $()$.  
Let 
$\level($a base type$)=0$, 
$\level(\sigma+\tau)$ $=$ $\level(\sigma\times\tau)$ 
$=$ $\max(\level(\sigma),\level(\tau))$,  
$\level(\sigma\to\tau)$ $=$ $\max(1+\level(\sigma),\level(\tau))$, 
and $\Ty_i^B = \set{\sigma\in\Ty^B \suchthat \level(\sigma)\leq i}$.
We call level-0 types 
\emph{ground types}.  
A type judgment $\Gamma\entails e\of\sigma$ asserts
that $e$ can be assigned type $\sigma$ under type context $\Gamma$,
where a \emph{type context} is a finite function from variables to types.

\emph{Algebraic Notions.} 
$\SetCat$ denotes the category of sets and total functions.  
Below we are mainly concerned with total functions and lower
type-levels, so $\SetCat$ suffices as the setting for the semantics of 
our programming formalisms.
Types are thus interpreted as sets where
coproduct ($+$), product ($\times$), and exponentiation ($\to$) have
their standard $\SetCat$-interpretations.  Let $\inj_i\of
A_i\to A_1+A_2$ ($i=1,2$) be the canonical coproduct injections and
$\proj_i\of A_1\times A_2 \to A_i$ ($i=1,2$) be the canonical
product projections.
A \emph{polynomial functor} is a functor inductively built from 
identity and functors and coproducts and products,
e.g.,
$F_{0}\, X = \Unit + (\Nat \times X)$ with 
$F_0\, f = \id_{\Unit} + (\id_{\Nat}\times f)$, where 
$\Nat$ is the type of natural numbers introduced below in 
Example~\ref{ex:nattree}.\footnote{%
  Other authors (e.g., \cite{rutten:coalg}) use broader 
  notions of polynomial functor.} 
The constant-objects in our polynomial functors will always be
types.  \emph{Convention:} For $F$, a polynomial function given by
$F\,X = e$, and $\sigma$, a type, read $F \sigma$ as
the type $e[X\gets \sigma]$.
E.g., $F_0\,\Nat$ = $\Unit + \Nat \times \Nat$.

\emph{The Base Formalism.}
This paper's programming formalism are built atop $L$, a standard,
simply-typed, call-by-value lambda calculus.  The $L$-types are
$\Ty^{\emptyset}$.  Figs.~\ref{fig:core:syn}
and~\ref{fig:core:typing} give $L$'s syntax and typing rules.  We
use the standard syntactic sugar:
\begin{inparaenum}[\it (i)]
\item 
  $\Let x_1=e_1;\dots; x_m=e_m \In e_0$ 
  $\equiv$
  $(\lam{x_1,\dots,x_m}e_0)\;e_1\;\dots\;e_m$ and 
\item 
  $\Letstar x_{1}=e_{1}; \dots; \allowbreak 
  x_{m}= e_{m} \In e_{0}$
  $\equiv$
  $\Let x_1=e_1\In\allowbreak 
  (\dots \; (\Let x_m=e_m\In e_0)\dots)$.
\end{inparaenum}

\emph{Semantics.} 
The denotational semantics of $L$ is standard.  As $\Unit$ is the
sole base type of $\Ty^\emptyset$, for each
$\sigma\in\Ty_1^\emptyset$, \ $\semantics{\sigma}$ is a finite set.
$L$'s operational semantics is also fairly standard as specified by
the evaluation relation, $\yields$, described in
Fig.~\ref{fig:core:eval}.  \emph{Terminology:} An \emph{evaluation
  relation} relates closures to values.  A
\emph{closure}~$(\Gamma\entails e\of\tau)\theta$ consists of a term
$\Gamma\entails e\of\tau$ and an environment~$\theta$ for
$\Gamma\entails e\of\tau$.  (We write $e\theta$ for $(\Gamma\entails
e\of\tau)\theta$ when $e$'s typing is understood.)  An
\emph{environment}~$\theta$ for $\Gamma\entails e\of\tau$ is a
finite map from variables to values with
$\fv(e)\subseteq\dom(\theta)\subseteq\dom(\Gamma)$ and, for each
$x\in\dom(\theta)$, $\theta(x)$ is a type-$\Gamma(x)$ value.
A \emph{value}~$z\theta$ is a closure in which~$z$ (the \emph{value
  term}) is either an abstraction or else an internal representation
of $()$ or $\mathop{\inj_i} v_i$ or $(v_1,v_2)$, where $v_1$ and
$v_2$ are value terms.  By \emph{internal representation} we mean
the ``machine'' representation of value terms, the details of which
are not important for $L$, but vital for the $RS^-$ and $RS$
formalisms below.  

\section{Structural Recursions} 	

\textbf{The Classical Case.} \enspace
We extend $L$ to $S^-$, a formalism that computes,
roughly, G\"{o}del's primitive recursive functionals
\cite{Longley:notions:1} over inductively-defined data
types.  Later we  introduce $\RSmi$, a ramified, ``feasible''
version of $S^-$.  Fig.~\ref{fig:S-} gives the revised raw syntax
\eqref{e:Ssyn}, typing rules ($c_\tau$-\emph{I},
$d_\tau$-\emph{I}, $\fold_\tau$-\emph{I}) and evaluation rules
(\emph{Const}$_\tau$, \emph{Destr}$_\tau$,  \emph{Fold}$_\tau$)
for $S^-$.  A declaration, $\data \tau = \lfp{t}\sigma$, introduces
a data-type $\tau$.
The polynomial functor $F_\tau t = \sigma$ is called $\tau$'s
\emph{signature functor}.  The declaration also implicitly
introduces: 
  $\tau$'s \emph{constructor function} 
        $c_\tau\of F_\tau\tau\to\tau$,  
  $\tau$'s \emph{destructor  function}  
        $d_\tau\of \tau \to F_\tau \tau$,
        and
  $\tau$'s \emph{recursor} $\fold_\tau \of (\forall
        \sigma)[(F_\tau\sigma \to \sigma) \to \tau \to \sigma]$.
We require that the $\sigma$ in $\data \tau = \lfp{t}\sigma$ be a
ground type with constituent base types are drawn from $t$, $\Unit$,
and previously declared types.  Semantically, the data type $\tau$
is the \emph{least fixed point} of $F_\tau$: it is a
smallest set $X$ isomorphic to $F_\tau(X)$, where $c_\tau$ and
$d_\tau$ witness this isomorphism.
It is standard that polynomial functors have 
such least fixed points.
In examples we use syntactically-sugared versions of $\data \tau =
\lfp{t}\sigma$ of the form:
$\data \tau = C_1 \Of \sigma_1 \alt \dots \alt C_n \Of \sigma_n$,
where $F_\tau(\tau) = \sigma_1+\sigma_2+\dots+\sigma_n$ and, for each $i$, 
if $\sigma_i=\Unit$, 
   then $C_i$ $\equiv$ $c_\tau\circ\inj_i^n() : \tau$ and 
if $\sigma_i\not=\Unit$, 
   then $C_i$ $\equiv$
   $c_\tau\circ\inj_i^n : \sigma_i\to\tau$.\footnote{%
  For $1\leq i < n$, define: 
  $\inj_i^n = \inj_1\circ \inj_2^{(i-1)}$ and
  $\inj_n^n = \inj_2^{(n-1)}$.
  Also, define: $\inj_1^1 = \id$.   
}   
Type-$\tau$ data can then be identified with the elements of the
free algebra over the sugared constructors $C_1,\dots,C_n$ and the
other constituent data-types' constructors.

\begin{figure}[t]
\begin{minipage}{\textwidth}\small
\begin{gather}
\label{e:Ssyn}
\Dcl \;\is\; \data T =  \lfp{T'} \Ty_0^T
\Quad{3}
P \;\is\; \declare\; \Dcl \;\big(;\;\Dcl\big)^* \;\In \; E
\\[0ex]
\notag
  \rulelabel{$c_\tau$-I}
  \irule{\Gamma\entails e\of F\tau}{\Gamma\entails (c_\tau\,e)\of\tau}
  \Quad2
  \rulelabel{$d_\tau$-I}
  \irule{\Gamma\entails e\of\tau}{\Gamma\entails (d_\tau\,e) \of F\tau}
  \Quad2
  \rulelabel{$\fold_\tau$-I}
  \irule{\Gamma \entails f\of F\sigma\to\sigma
  \Quad{1} 
  \Gamma\entails e\of\tau}{\Gamma\entails \fold_\tau f \,e\of \sigma} 
  \\[0ex]
\notag
  \rulelabel{Const$_\tau$}
  \irule{e\theta \yields v\theta'}%
  {(c_{\tau}\,e)\theta \yields (\underline{c}_{\tau}\,v)\theta'} 
  \Quad3
  \rulelabel{Destr$_\tau$}
  \irule{e\theta \yields (\underline{c}_\tau\,v)\theta'}%
        {(d_{\tau}\,e)\theta \yields v\theta'} 
  \\[0.5ex]
\notag
  \rulelabel{Fold$_\tau$}
  \irule{e\theta\yields (\underline{c}_\tau\,v)\theta_1
  \Quad1 
  f(F\,(\lam{x}(\fold_\tau f\,x))\,y) 
      \theta[y\mapsto v\theta_1]\yields v'\theta'}%
  {(\fold_\tau f\,e)\theta \yields v'\theta'}
  \sidecond{\hbox{$x$ and $y$ are fresh}}
\end{gather}
\caption{Extensions for $S^-$, where $\tau$  is   
is a data-type with signature functor $F$.}
\label{fig:S-}
\end{minipage}
\end{figure}

\begin{example}\label{ex:nattree} \
  The declaration,
  $\data \Nat = \Zero \Of \Unit \alt \Succ \Of \Nat$,  
  introduces the type $\Nat$ with signature functor $F_\Nat X =
  \Unit + X$ and sugared constructors $\Zero\of\Nat$ and
  $\Succ\of\Nat \to\Nat$.  Type-$\Nat$ data thus corresponds to the
  terms of the free algebra over $\Zero$ and $\Succ$, i.e., $\Zero$,
  $\Succ(\Zero)$, $\Succ(\Succ(\Zero))$, etc.
\end{example}
\begin{example}
  The declaration,
   $\data \Tree = \Leaf \Of \Unit \alt \Fork \Of \Tree\times \Tree$, 
   introduces the type $\Tree$ with signature functor $F_\Tree X =
   \Unit + X \times X$ and sugared constructors $\Leaf\of\Tree$ and
   $\Fork\of\Tree\times\Tree\to\Tree$.  Type-$\Tree$ data thus
   corresponds to the terms of the free algebra over $\Leaf$ and
   $\Fork$, i.e., $\Leaf$, $\Fork(\Leaf,\Leaf)$,
   $\Fork(\Fork(\Leaf,\Leaf),\Leaf)$, etc.
\end{example}

The recursor for type-$\tau$ data, $\fold_\tau$, has its operational
semantics given by Fig.~\ref{fig:S-}'s 
\emph{Fold$_\tau$} rule\footnote{%
          In the rule \emph{Fold}$_\tau$, the use of $F$ should be
          read as shorthand for a $\lambda$-term that
          expresses,  in $S^-$, the polynomial function $F$ 
          (specialized to the appropriate types).}
and satisfies: 
$(\fold_\tau g)\circ c_\tau = g \circ F(\fold_\tau g)$.
This last equation expresses structural (aka \emph{primitive}) 
recursion over $\tau$.  For example, given  an 
$f\of F_{\Tree}\Nat \to\Nat$ with 
$f(\iota_1())=\Zero$ and $f(\iota_2(x,y))=\Succ(\max(x,y))$,
then $(\fold_{\Tree} f \,t)$ computes the height of $\Tree$ $t$. 
As the  $g$ in $(\fold_{\tau} g \,x)$ can be of any positive
type level, one can show that $S^-$ computes a
version of G\"{o}del's primitive recursive functionals. 
To rein in the  power of $\fold$-recursions to express just low 
complexity computations,  we apply a standard tool of implicit 
complexity, ramification. 
First, however, we need to consider 
how data is represented 
and how the size of a representation is measured.

\Topic{Representation, Size, and Memoization}
\emph{Representing Data.} 
Our internal representation of data
follows standard practice in implementations of
functional languages.  Each invocation of a constructor function:
\begin{inparaenum}[\it (i)]
  \item 
    allocates a fresh \emph{cons-cell} that stores the values of 
    the invocation's arguments and 
  \item 
    returns, as its value, a pointer to this new cons-cell.
\end{inparaenum}    
\textbf{N.B.} The product and coproduct constructors also create
cons-cells. 
%
As our formalism is purely functional, it follows that all data is
represented by directed acyclic graphs (\emph{dag}s) on cons-cells.

\emph{Measuring The Size of Data Representations.}
A data-representation's \emph{size} is simply the number of
\emph{data} cons-cells in the representation.  For example,
consider:
\begin{gather}\label{e:tn}
  \Letstar \, t_0 = \Leaf;\;
           t_1=\Fork(t_0,t_0);\,
	         \dots\,;\;
           t_n=\Fork(t_{n-1},t_{n-1}) 
           \,\In\, t_n
\end{gather}
The size of $t_n$'s representation is $n+1$ (one \Leaf-cell and $n$
\Fork-cells).  This notion of size depends on the \emph{operational}
semantics.  Denotationally, $t_n$ names a proper tree which is also
named by $t_n'$, a sized-$(2^{n+1}-1)$ $\Tree$ consisting of $2^n$
\Leaf-cells and $(2^n-1)$ \Fork-cells.\footnote{ For simplicity, we
  do \emph{not} count the cons-cells of product and coproduct
  constructors in representations as the asymptotics
  are the same whether we count these or not.}

\begin{definition}
  \label{d:ap:size}
  Suppose $e_0\theta, \dots, e_k\theta$ are ground-type
  closures.
  The \emph{apparent size} of 
  $\set{e_0\theta, \dots, e_k\theta}$
  (written: 
  $\asize{e_0, \dots, e_k}\theta$) is
  the number of data cons-cells in the representation the values of 
  $e_0\theta,\dots,e_k\theta$.
  {(\textbf{N.B.} This takes account of sharing.  
  E.g., if  
  $\theta$ is the environment 
  in force in the body of \eqref{e:tn}, then 
  $\asize{t_0,\dots,t_n}\,\theta = n+1$.)}
\end{definition}

\emph{Memoized Structural Recursions.}
Two of our goals for our feasible programming formalisms are:
\begin{inparaenum}[\it (i)] 
  \item 
    to have the {run-time} of programs to be polynomial-bounded in 
    {the size of the representations they compute over}; and
  \item      
    to have our programs to return equivalent results on 
    equivalent inputs (e.g., $t_n$ and $t_n'$ as above).  
\end{inparaenum}
These goals would seem to conflict given our conventions on
data-re\-pre\-sentions and sizes.  This is resolved via the standard
programming trick of memoization \cite{AbelsonSussman}.  Computing
$(\fold_\tau f\,x)$ can be treated as a linear programming problem
with $x$'s data representation as the underlying dag,  there is,
then, an exact match between the $\fold$-recursion's steps and $x$'s
cons-cells, moreover, the result of each step is stored for possible reuse
later in the recursion.  We assume that our structural-recursion
implementation uses memoization for just \emph{branching} data types
(e.g., $\Tree$); for nonbranching data-types (e.g., $\Nat$) it is
not needed.

\Topic{The Ramified Case}
$\RSmi$, our ramification of $S^-$, uses Bellantoni and Cook's
normal/safe distinction that splits data into two sorts:
\emph{normal} data that drive recursions and \emph{safe} data over
which recursions compute. E.g., in $(\fold_\tau g\,x)$ we want $x$
normal and $g\of$(safe data) $\to$ (safe data).  Typing constraints
enforce this distinction, which is roughly the idea behind
Bellantoni and Cook's $BC$ function algebra
\cite[\S5]{BellantoniCook} (and Leivant's formalism from
\cite{Leivant:FM2}), but \textbf{not} Bellantoni and Cook's better
known $B$ function algebra.
\emph{Normal types:}
The \emph{normal base types} consist of $\Unit$ and the types
directly introduced by $\data$-definitions.  The \emph{normal ground
  types} are the closure of the normal base types under $+$
and $\times$.  In $\data \tau = \lfp{t}\sigma$, we require that
$\sigma$ be normal.  A declaration $\data \tau = \lfp{t}\sigma$
introduces $c_\tau$ and $d_\tau$ as before, but $\fold_\tau$ is
replaced with $\folds{\tau}$ as explained shortly.
\emph{Safe types:} 
By convention, $\data \tau = \lfp{t}\sigma$ implicitly introduces a
parallel type $\safe\tau$. 
We extend -$\safe{}$ to all normal ground types by:
$\safe\Unit=\Unit$,
$\safe{(\sigma_1\times\sigma_2)}=\safe\sigma_1\times\safe\sigma_2$,
$\safe{(\sigma_1+\sigma_2)}=\safe\sigma_1+\safe\sigma_2$.  (Note:
$\Unit$ is the sole normal \emph{and} safe base type.)
$\safe\tau$ has constructor
$\safe{c_\tau}\of\safe{(F_\tau\tau)}\to\safe\tau$ and destructor
$\safe{d_\tau}\of\safe\tau\to\safe{(F_\tau\tau)}$.  In examples, we
use sugared constructors for $\safe\tau$, e.g., $\safe\Succ \of
\safe\Nat\to\safe\Nat$. 
The elements of $\safe\tau$ are essentially  ``safe'' copies of the
elements of $\tau$.  Let $\sfv(\Gamma\entails e\of\tau)
= \{\, x \in\fv(e) \suchthat \Gamma(x)$ is safe$\,\}$, which
we write as $\sfv(e)$ when the judgment is understood.

The new recursor for $\tau$-data,
$\folds{\tau}$, has the same operational semantics as $\fold_\tau$
(Fig.~\ref{fig:S-}: \emph{Fold$_\tau$}) and the same typing rule as
$\fold_\tau$ 
\emph{except} for the new 
side-condition, see
Fig.~\ref{fig:RS-}.
\emph{Examples:}
$\up_\tau = \lam{x}(\folds{\tau} \safe{s}_\tau\,x)\
\of\tau\to\safe\tau$
translates each type-$\tau$ datum to the corresponding type-$\safe\tau$
datum; 
$\mathit{plus} = 
\lam{x,y}(\Let g =\lam{z} \Case z \Of \,
(\iota_1 w) \Rightarrow y; 
(\iota_2 w)\Rightarrow (\safe\Succ w) 
\In \,(\folds{\Nat} g\, x)\of \Nat\to\safe\Nat\to\safe\Nat$ 
adds its  arguments and 
$\mathit{times} = \lam{x,y}(\Let h =\lam{z} \Case z \Of \,
(\iota_1 w) \Rightarrow (\safe\Succ \Zero); 
\allowbreak
(\iota_2 w)\Rightarrow (\mathit{plus}\,x\,w) 
\In\, (\folds{\Nat} h\, y)\of \Nat\to\Nat\to\safe\Nat$ multiplies
its  arguments. 

Ramified type systems have a perennial difficulty: certain natural
compositions can be untypable, e.g., $\mathit{cube} =
\lam{x}\mathit{times}\,x\,(\mathit{times}\,x\,x)$ fails to type
using the rules stated so-far.  As a mitigation, we introduce the
\emph{lower} typing rule (Fig.~\ref{fig:S-}) which is an adaptation
to $\lambda$-calculi of Bellantoni and Cook's Raising Rule
\cite{BellantoniCook}.
Using \emph{lower} on the $(\mathit{times}\,x\,x)$ subterm yields a
$\Nat\to\safe\Nat$ version of $\mathit{cube}$ and an second
application of \emph{lower} yields a $\Nat\to\Nat$ version.  When we
say a type-1 function is $\RSmi$-computable, we usually mean it is
computed by a type-$(\sigma\to\tau)$ $\RSmi$-term where both
$\sigma$ and $\tau$ are normal.

\begin{figure}[t]
\begin{minipage}{\textwidth}\small
\begin{gather*}
  \rulelabel{$\folds\tau$-I}
  \irule{\Gamma \entails f\of F\sigma\to\sigma
  \Quad{1} 
  \Gamma\entails e\of\tau}{\Gamma\entails \folds\tau f \,e\of \sigma} 
  \sidecond{\dagger}
  \Quad2
  \rulelabel{lower}
  \irule{\Gamma\entails e\of\safe\tau}{
  \Gamma\entails e\of\tau}
  \sidecond{\ddagger}
\end{gather*}
\caption{Key additions for $\RS^-$.  \quad
$(\dagger)$ $\tau$ is a normal and $\sigma$ is safe.
\quad 
$(\ddagger)$  $\sfv(e)=\emptyset$.}
\label{fig:RS-}
\end{minipage}
\end{figure}

$\RSmi$ is thus the modification of $S^-$ as sketched above
with one last change:
$+$-\emph{I} and  $\times$-\emph{I} now have the side-condition
that the component  types, $\sigma_1$ and $\sigma_2$, are 
both either normal ground types or safe ground types.
\emph{(Thus each ground-type $\RSmi$-term is of either of normal- 
or safe-type. This simplifies size bounds.)}

\emph{Poly-Heap Size Bounds.}  Bellantoni and Cook proved poly-max
size bounds for their formalisms, e.g., if $e$ is a base-type
(string-valued) $BC$-expression, then, for all $\theta$, \
$|e\theta| \leq (p + \max_{y\in\sfv(e)}|y|)\theta$, where $p$ is a
\emph{normal polynomial}, i.e., $p$ is is polynomial over $\{\,|x|
\suchthat x\in\fv(e)$ \& $x$ has a normal type$\,\}$.  Because of
sharing we replace poly-max with \emph{poly-heap} bounds, i.e.,
those of the form $p+\asize{y_1,\dots,y_n}$ (recall
Definition~\ref{d:ap:size}) where $p$ is a normal polynomial and
$\set{y_1,\dots,y_n}=\sfv(e)$.  (\emph{Convention:} 
We write bounds
as $\asize{e}\leq p+\asize{y_1,\dots,y_n}$, keeping the universal
quantification over $\theta$ implicit
and in place of $\asize{y_1,\dots,y_n}$ we  write 
$\asize{\sfv(e)}$.

\begin{theorem}[$\RSmi$ Poly-Heap Size-Boundness] \label{t:RS-:size} 
  Given an $\RSmi$ judgment $\Gamma\entails e\of \tau$  in which
  $\tau$ and each $\sigma\in\image(\Gamma)$ is a ground type, one can
  effectively find a normal polynomial $p$ with
$\asize{e} \leq p + \asize{\sfv(e)}$.
\end{theorem}

A partial proof of Theorem~\ref{t:RS-:size} is given in the Technical 
Appendex.
$\RSmi$ also satisfies \emph{poly-cost boundness} (the computation
tree of $\Gamma\entails e\of \tau$ has a poly-size bound over
$\set{|x|\suchthat x\in\dom(\Gamma)}$) and 
\emph{poly-completeness} (for a suitable model of
computation and cost, $\RSmi$ can compute all poly-time computable
type-1 functions).  For want of space we omit these results, but
their proofs are similar to analogous
results in \cite{DR:ATS:LMCS,DannerRoyer:2algs}.
\textbf{N.B.}  While the 
completeness result's proof is standard, the result itself is a
little subtitle.  Typically, complexity classes concern the purely
representational level and not extensionality constraints imposed by
the things represented.  In contrast, the $\RSmi$-computable
$(\Tree\to\Tree)$-functions form a nonstandard complexity class: all
the poly-time (in the dag-size) computable functions over
$\Tree$-representations which functions are  extensional with
respect to $\Tree$-data.  
\emph{Type Restricted $\RSmi$.}  
Let $\RSmi[\Nat]$ be the restriction of $\RSmi$ to terms with types
in $\Ty^{\set{\Nat,\safe\Nat}}\!$.
It follows from \cite{bellantoni:thesis,Leivant:FM2} 
that the $\RSmi[\Nat]$-computable
$(\Nat\times\cdots\times\Nat\to\Nat)$-functions = $\mathcal{E}_2$,
the second Grzegorczyk class (aka, the linear-space computable
functions).  $\mathcal{E}_2$ plays a key r\^{o}le in ``two-sorted
complexity'' characterizations \cite[Chapter 4]{CookNguyen:book}.  We shall
make similar use of it below.

\section{Structural Corecursions}\label{S:codata}

\textbf{The Classical Case.} \enspace We extend $S^-$ to $S$, a
formalism that computes, roughly, G\"{o}del's primitive recursive
functionals over inductively- and coinductively-defined data.
$\RSi$ will be our ramified, ``feasible'' version of $S$.
Fig.~\ref{fig:S} gives the revised syntax \eqref{e:S} and evaluation
rules (\emph{Destr}$'_\tau$, \emph{Unfold$_\tau$}).  The typing
rules for $\co{c}_\tau$, $\co{d}_\tau$, and $\unfold_\tau$ are given
implicitly below.  A declaration, $\codata\tau=\gfp{t}\sigma$,
introduces a codata-type $\tau$.  The polynomial functor $F_\tau t =
\sigma$ is called $\tau$'s \emph{signature functor}.  The
declaration also implicitly introduces:
  	    $\tau$'s \emph{constructor function} 
        $\co{c}_\tau\of F_\tau\tau\to\tau$,  
        $\tau$'s \emph{destructor  function}  
        $\co{d}_\tau\of \tau \to F_\tau \tau$,
        and
        $\tau$'s \emph{corecursor} $\unfold_\tau \of (\forall
        \sigma)[(\sigma \to F\sigma) \to \sigma \to \tau]$.
The $\sigma$ in 
$\codata\tau= \gfp{t}\sigma$ must be a ground type with constituent
base types drawn from $t$, $\Unit$, and previously declared types.
Type-$\tau$'s corecursor, $\unfold_\tau$, has its operational
semantics given by Fig.~\ref{fig:S}'s \emph{Unfold}$_\tau$-rule and
satisfies:
$\co{d}_\tau\circ (\unfold_\tau\,f) = F_\tau (\unfold_\tau\,f)\circ f$.
\textbf{N.B.} Codata constructors and unfolds are lazy:
$\co{c}_\tau$- and $\unfold_\tau$-expression are values and hence 
are not evaluated unless forced by a $\co{d}_\tau$-application 
per \emph{Destr}$'_\tau$ and \emph{Unfold$_\tau$}.
Semantically, a codata type $\tau$ is the \emph{greatest fixed point} 
of $F_\tau$: it is a largest set $X$ isomorphic 
to $F_\tau(X)$, where $\widehat{d}_\tau$ and $\widehat{c}_\tau$
witness the isomorphism.  
Polynomial $\SetCat$-functors are know to have 
such greatest fixed points \cite[Theorem~10.1]{rutten:coalg}. 
In examples, we use sugared 
$\codata$-declarations along the lines of the sugared
$\data$-declarations.   

\begin{figure}[t]
\begin{minipage}{\textwidth}\small
\begin{gather}
\label{e:S}
Dcl \;\;\is\;\; \data\; T = \lfp{X}\Ty_0
  \synsep \codata\; T = \gfp{X}\Ty_0
\\[0ex]
\notag
  \rulelabel{Destr$'_\tau$}
  \irule{e\theta \yields v'\theta'}{(\co{d}_{\tau}\,(\co{c}_{\tau}\,e))\theta \yields 
  v'\theta'} 
  \Quad3
  \rulelabel{Unfold$_\tau$}
  \irule{(F(\unfold_\tau f))(f \, e)\theta\yields v'\theta'}%
  {\co{d}_{\tau}(\unfold_\sigma f\,e)\theta \yields v'\theta'}
\end{gather}
\caption{Key Additions for $S$.}
\label{fig:S}
\end{minipage}
\end{figure}

\begin{example}\label{ex:nats}
  The declaration, $\codata \Nats$ $=$ $\Cons\Of\Nat\times\Nats$,
  introduces the type $\Nats$ with signature functor 
  $F_{\Nats} X = \Nat \times X$ 
  and  constructor $\Cons\of \Nat\times\Nats\to\Nats$.
  Each element of $\Nats$ corresponds to an infinite 
  sequence of $\Nat$'s.
  Given an $f\of\Nat\to\Nat$, let 
  $\ms = \unfold_\Nats\,(\lam{x}\Case x \Of 
  (\iota_1 y) \Rightarrow (f\,\Zero,\allowbreak \Succ\Zero)$;
  $(\iota_2 y) \Rightarrow (f\,(\Succ y),\Succ (\Succ y))$,
  so $\ms \equiv$ the sequence $f(0)$, $f(1)$, $f(2), \cdots\;$.
  Given an $\ns\of\Nats$, let 
  $g=\lam{n}\pi_1(\fold_\Nat (\lam{x}\Case x \Of
   (\iota_1 y) \Rightarrow \ns;$ 
  $(\iota_2 y) \Rightarrow (\co{d}_\Nats\,y)) \; n)$, 
   so $g(n) = $ the $n$th $\Nat$ in  $\ns$'s sequence. 
\end{example}

As the above shows, codata are really higher-type objects.  To help
analyze this, define
a rank-0 type is at type with no constituent codata types, and
a rank-$(k+1)$ type is a type with constituent codata types of maximum rank $k$.
E.g., $\Nat$, $\Tree$, and $\Nats$ are rank 0 and a stream of
$\Nats$ is rank 1.  Let $S_k$ be the restriction of $S$ to types of
levels $\leq 1$ and 
ranks $\leq k$.
Not surprisingly, the $S_k$-functions of types
$\Nat\times\dots\times\Nat\to\Nat$ correspond to P\'eter's
$(k+1)$-primitive recursive functions \cite{Longley:notions:1}.
We shall show how normal/safe ramification can rein in the power of
these corecursions.  First, we consider codata representations and
their size.

%



\Topic{Representation and Size}
A type-$\tau$ codatum $x$ is represented via lazy $\co{c}_\tau$-
and/or $\unfold_\tau$-expressions; if we probe $x$ with ever-longer
series of destructor applications, a possibly infinite structure
unfurls.  A codatum is thus a function-like object that must be
queried (via destructor applications) to be computed over.  To
measure codata-size we adapt Kapron and Cook's notion of the length
of a type-1 function \cite{KapronCook:mach}.  Measuring just rank-0
codata suffices for this paper.

\begin{definition}  \label{d:ob:size} \
Suppose $e\theta$ is of type $\tau$, a rank-0 codata-type.
  
  \begin{asparaenum}[(a)]
  \item 
    The \emph{apparent size} of $e\theta$ (written:
    $\asize{e}\theta$) is 1.
  
  \item 
    The \emph{observed size} of $e\theta$ (written:
    $\osize{e}\theta$) is the function over natural numbers:
    $n\mapsto\max(\{\, \asize{\vec{d}(e)}\theta \suchthat \vec{d}$
    varies over sequences of compositions of destructors with (i)
    $\vec{d}(e)$ type correct and (ii) at most $n$ occurrences of
    $\co{d}_\tau\,\})$.

  
  \end{asparaenum}
\end{definition}

Roughly, $(\osize{e}\theta)(n)$ is the maximum apparent-size of the
data in $\tau$-cons-cells along any path from the head of $e\theta$
that includes at most $n$ type-$\tau$ links.  \emph{Example:} For
$\ns$ of Example~\ref{ex:nats},
$(\osize{\ns}\theta)(n)=1+\max_{i<n}$(the $i$th element of $\ns$'s
sequence).

\Topic{The Ramified Case}
$\RSi$, our ramification of $S$, extends the normal/safe distinction
to codata.  \emph{Key Points:} As the value of
$(\unfold_\tau g)$ is the result of a (co)recursion, it should be
safe,  as $g$ gives the computation step, we should have
$g\of\mathrm{safe}\to\mathrm{safe}$, and as $\unfold$'s are lazy, 
destructs drive the computation.
 \emph{Normal and Safe
  Types:} First, we bring in all the
$\RSmi$ conventions to this setting to ramify data.
Second, a declaration
$\codata\tau=\gfp{t}\sigma$ introduces the normal type $\tau$ with
constructor $\co{c}_\tau$ and destructor $\co{d}_\tau$ as before, a
safe type $\safe\tau$ with constructor $\safe{\hat{c}}_{\tau}\of
\safe{(F_\tau \tau)}\to\safe\tau$ and destructor
$\safe{\co{d}}_{\tau}\of \safe\tau\to\safe{(F_\tau \tau)}$, and
$\unfolds{\tau}\of(\forall \hbox{ normal }\sigma)[ (\safe\sigma\to
\safe{(F_\tau\sigma)})\to\safe\sigma \to \safe\tau]$ where $\unfolds{\tau}$ has
the same operational semantics as $\unfold_\tau$.  \emph{Example:}
Replace $\unfold_\Nats$, $\Zero$,  $\Succ$, and $f\of\Nat\to\Nat$
with $\unfolds{\Nats}$, $\safe\Zero$,  $\safe\Succ$, and
$f\of\safe\Nat\to\safe\Nat$ in
Example~\ref{ex:nats}'s definition of $\ms$, then $\ms$ can be
assigned assigned type $\safe\Nats$.  
\textbf{N.B.}
Given an $\RSmi$-computable $f\of\Nat\to\Nat$, there may 
\emph{not} be an $\RSi$-definable analogue of $\ms$ from
Example~\ref{ex:nats}.

\emph{Poly-Heap Size Bounds.} 
To adapt poly-heap bounds to take account of observed sizes we use
Kapron and Cook's notion of second-order polynomials
\cite{KapronCook:mach}; these are roughly ordinary polynomials with
applied type-1 function symbols included (e.g.,
$x^2+f(y+2)$).  Now $\asize{e} \leq p+\asize{\sfv(e)}$ 
is a \emph{poly-heap bound on apparent size} when $p$ is a 
\emph{normal second-order polynomial} 
(i.e., over $\{\,\asize{x} \suchthat \Gamma(x)$ is normal$\,\}$ 
and $\{\,\osize{x} \suchthat \Gamma(x)$ is a normal codata type$\,\}$)
and 
$\osize{e} \leq \lam{n}(p+\asize{\sfv(e)})$ 
is a \emph{poly-heap bound on observed size} where now $p$
can have $n$
as a type-0 variable.

\begin{theorem}[$\RSi$ Poly-Heap Size-Boundness] \label{t:RS:size}
  For an $\RSi$-judgment $\Gamma\entails e\of \tau$ where
  $\tau$ and each $\sigma\in\image(\Gamma)$ is a ground type, one
  can effectively find a normal second-order polynomial $p$ such
  that, if $\tau$ is a data-type, then $\asize{e} \leq p +
  \asize{\sfv(e)}$ and, if $\tau$ is a codata-type, then
  $\osize{e} \leq \lam{n}(p + \asize{\sfv(e)})$.
\end{theorem}

$\RSi$ satisfies appropriate \emph{poly-cost boundness}
and \emph{poly-completeness} properties with proofs similar to the
analogous (type-2) results in \cite{DR:ATS:LMCS,DannerRoyer:2algs};
but, as with $\RSmi$, we have not the space to describe, much less
prove, these results.

\emph{Type Restricted  $\RSi$.} 
Let $\RS_1(\sigma\to\tau\confined B) =$ the functions of type
$\sigma\to\tau$ computable by $\RS_1$-terms with types from
$\Ty^{B'}$ where $B'=$ the normal and safe versions of the base 
types occurring in $\sigma$, $\tau$, and $B$.
For
$\data \Bit = \Nought \alt \One$ and 
$\codata \Stream = \Cons\Of\Bit\times\Stream$,
one can show:
\begin{inparaenum}[\it (i)]
  \item 
	$\RSi(\,\Unit\to\Stream\confined\emptyset)$ =
  	$\omega$-regular languages,
  \item 
	$\RSi(\,\Stream\to\Stream\confined\emptyset)$ =
  	finite-state stream maps,
  \item \label{i:logstr}
    $\RSi(\,\Unit\to\Stream\confined\,\set{\Nat})$ =
    logspace streams, and
  \item \label{i:logtrans}
    $\RSi(\,\Stream\to\Stream\confined\,\set{\Nat})$ =
    logspace stream-functions.
\end{inparaenum}
(In \eqref{i:logstr} and \eqref{i:logtrans}, $\Nat$ plays the
r\^{o}le of counter/pointer type as it does in the two-sorted
characterizations considered in \cite{CookNguyen:book}.)


\section{Conclusions}\label{S:fini}

$\RSi$ characterizes a notion of poly-time computation over data 
and codata. As a formalism, $\RSi$ is not much more complicated  
than  the original ones of Bellantoni and Cook
\cite{BellantoniCook} and Leivant \cite{Leivant:FM2}, although
a few of  $\RSi$'s additions involve  subtleties.  
The above work suggests many paths for exploration.
Here, briefly, are a few.

\emph{Pola vs.~$\RSi$.}
Pola restricts sharing for its notion of poly-time
over data and codata. $\RSi$ essentially \emph{forces} sharing 
to obtain its notion of poly-time over data and codata. 
How different are these two notions?  Can one notion
``simulate'' the other in some reasonable sense?  Is there
a good notion of poly-time over data and codata that sits above
both the Pola and $\RSi$ notions?

\emph{Higher-types.} 
Higher-type functions over data-realm and higher-rank streams and
trees in the codata-realm are roughly two different perspectives 
on the same thing.  In investigating true higher-type
extensions of $\RSi$, having these two views may help puzzling out
sensible approaches to higher-type feasibility.

\emph{Programming in $\RSi$ is clumsy.}  
One problem is that  $\RSi$-recursions carry out their computations 
using $\mathrm{safe}\to\mathrm{safe}$ functions, but there are very few
of these that have \emph{closed} definitions in $\RSi$. E.g.,
there is no closed $\RSi$-function that gives the
$\safe\Nat$-maximum of two $\safe\Nat$-values, even though adding
such a function would be a complexity-theoretic conservative 
extension.
Based on an insight first pointed out and studied by
Hofmann \cite{Hofmann03}, any polynomial-time computable
$f\of\mathrm{safe}\to\mathrm{safe}$ with $\asize{f(x)}\leq
\asize{x}$ for all $x$, would be a similarly conservative extension to
$\RSi$.  Finding a simple scheme to add to $\RSi$ that allows the definition
of more such functions over data (and the dual notion,
$\osize{x}\leq \osize{f(x)}$, for functions over codata) is a nice
problem.

%
%
%

\bibliography{ops}

\newpage
\appendix
\section*{Technical Appendix}

\setcounter{figure}{0}
\renewcommand{\thefigure}{A.\arabic{figure}}
\setcounter{theorem}{0}
\renewcommand{\thetheorem}{A.\arabic{theorem}}
\setcounter{lemma}{0}
\renewcommand{\thelemma}{A.\arabic{lemma}}

\subsubsection{Notes}

\begin{figure}[t]
\begin{minipage}{\textwidth}\small
\begin{align*}
  E \;\;&\is\;\; 
        X
        \synsep (E_1 \; E_2)
        \synsep (\lam{X}E)
	    \synsep ()
        \synsep (E_1,E_2)
        \synsep (\proj_1 E)
        \synsep (\proj_2 E)
		\\
        &
        \Quad1
        \synsep (\inj_1 E)
        \synsep (\inj_2 E)
        \synsep \Case E_0  \Of  
        (\inj_1 X_1) \Rightarrow E_1 ;\;
        (\inj_2 X_2) \Rightarrow E_2
\end{align*}
\caption{$L$ raw syntax, where $X\is$ identifiers.}\label{fig:core:syn}
\end{minipage}
\begin{minipage}{\textwidth}\small
\begin{gather*}
    \rulelabel{Id-I}
    \irule{ 
                }{\Gamma,\,x\of\tau\entails x\of\tau} 
	\Quad{2}
    \rulelabel{$\to$-I}
    \irule{ 
                \Gamma,\,x\of\zeta\entails e\of\tau}{
        \Gamma\entails (\lam{x} e) \of\zeta\to\tau} 
        \Quad{2}
    \rulelabel{$\to$-E}\;
    \irule{ 
                \Gamma\entails e_0\of\zeta\to\tau 
                \Quad{1.5}
        			\Gamma\entails e_1\of\zeta}{
                \Gamma\entails (e_0\;e_1)\of \tau}
	\\[0ex] 
   \rulelabel{$\Unit$-I}
   \irule{}{\Gamma\entails () \of \Unit}
	\Quad{2}
    \rulelabel{$\times$-I}
    \irule{ 
	            \Gamma\entails e_1\of\sigma_1
	            \Quad{1.5}
	            \Gamma\entails e_2\of\sigma_2}{
                \Gamma\entails (e_1,e_2)
                \of\sigma_1\times\sigma_2} 
	\Quad{2}
    \rulelabel{$\times$-E$_i$}
    \irule{ 
	            \Gamma\entails e\of\sigma_{1}\times\sigma_2}{
                \Gamma\entails (\proj_i e)\of\sigma_i} 
					\sidecond{\dagger}
	\\[0ex] 
    \rulelabel{$+$-I$_i$}
    \irule{ 
	            \Gamma\entails e\of\sigma_{i}}{
                \Gamma\entails (\inj_i e)\of\sigma_1+\sigma_2} 
					\sidecond{\dagger}
    \Quad{2}
    \rulelabel{$+$-E}
    \irule{
	\Gamma\entails e_0\of \sigma_1+\sigma_2
    \Quad{1.5}             
    \set{\Gamma,x_i\of\sigma_i\entails e_i\of \tau}_{i=1,2}}%
    {\Gamma\entails (\Case e_0 \Of 
        (\inj_1 x_1) \Rightarrow e_1 ;\;
        (\inj_2 x_2) \Rightarrow e_2)\of\tau}
\end{gather*}
\caption{$L$ typing rules.\quad  ($\dagger$) $i=1,2$.}
  \label{fig:core:typing}
\end{minipage}
\begin{minipage}{\textwidth}\small
\begin{gather*}
  \rulelabel{Val}  
  \irule{}{v\theta \yields v\theta}
  \sidecond{\parbox{1cm}{\centering $v\theta$ is a value}}
  \Quad{1.5}
  \rulelabel{Env} 
  \irule{}{%
        x\theta\yields v\theta'}  
        \sidecond{\theta(x) = v\theta'}
\\
  \rulelabel{$\lambda$-App} 
  \irule{e_0\theta\yields(\lam{x}e')\theta_0 \Quad{1}
         e_0\theta_1\yields v_1\theta_1\Quad{1}
         e'\theta_0[x\mapsto v_1\theta_1]\yields v\theta' 
         }{%
         (e_0\;e_1)\theta\yields v\theta'}
\\
  \rulelabel{Inj$_i$} 
  \irule{e\theta\yields v\theta'}{%
         (\inj_i e)\theta\yields (\uiota_i v)\theta'}
					\sidecond{\dagger}
\Quad{1}
  \rulelabel{Case} 
  \irule{e_0\theta\yields (\uiota_i v_i)\theta_i\Quad{1.25} 
         e_i\theta[x_i\mapsto v_i\theta_i]\yields v\theta'}{%
         (\Case e_0  \Of  
        (\inj_1 x_1) \Rightarrow e_1 ;\;
        (\inj_2 x_2) \Rightarrow e_2)\theta\yields
         v\theta'}
					\sidecond{\dagger}
\\
  \rulelabel{Unit} 
  \irule{}{()\theta \yields \underline{()}\theta}
  \Quad{1.5}
  \rulelabel{Pair} 
  \irule{e_1\theta\yields v_1\theta_1 \Quad{1.25}
         e_2'\theta_1\yields v_2\theta_2}{%
         (e_1,e_2)\theta \yields \upair{v_1}{v_2}\theta_2}
         \sidecond{\ddagger}         
  \Quad{1.5}
  \rulelabel{Proj$_i$} 
  \irule{e\theta\yields \upair{v_1}{v_2}\theta'}{%
         (\proj_i e)\theta \yields v_i\theta'}
					\sidecond{\dagger}
\end{gather*}
\caption{$L$ evaluation rules. 
\   ($\dagger$) $i=1,2$.
\  ($\ddagger$) $e_2'\equiv_\alpha e_2$, but $e_2'$ avoids
clashes with $\theta_1$.
}
\label{fig:core:eval}
\end{minipage}
\end{figure}
\begin{enumerate}
 \item 
    The ``S'' in $S^-$ and $S$ stands for \emph{structure} and 
    the ``R'' in $\RSmi$ and $\RSi$ stands for \emph{ramified}.
 \item 
    Internal 
    representations of constructors are {underlined}
    as in Figs.~\ref{fig:S-} and~\ref{fig:core:eval}.
 \item 
   \emph{The side-condition of \emph{Pair}-rule in 
   Fig.~\ref{fig:core:eval}.}
   If $e_1\theta\yields v_1\theta_1$ and 
   $e_2\theta\yields v_2\theta_2$, then $\theta_1$ and $\theta_2$ 
   may  be inconsistent. Hence in \emph{Pair}, $e_2$ is alpha-reduced
   to $e_2'$ so that the $e_1$- and $e_2'$-evaluations introduce 
   distinct
   variables into their value's environments.
  \item 
    The unsugared version of $\Succ (\Succ \Zero)$ is
	$c_\Nat (\iota_2 (c_\Nat (\iota_2 (c_\Nat (\iota_1 ())))))$.

 \item 
   \emph{Call-by-value and growth.}
   Note that for $e'$ of ground-type,    
   $\asize{((\lam{x}\Fork x\,x)\,e')}$ $=$ $1+\asize{e'}$
   because, by the call-by-value semantics, $e'$ is evaluated to a 
   value $v\theta$ (i.e., a reference to a data-representation) which 
   becomes the value of $x$ used in $\Fork x\,x$. This is explicit in 
   our closure-based evaluation semantics, since this expression 
   evaluates to $\Fork(x, x)[x\mapsto v\theta]$.
 \item  
   \emph{Dodging exponential growth.}
   \emph{If} one could define a function $f\of\Nat\to\Tree$ such that
   $f \Zero = \Leaf$ and $f (\Succ x) = \Fork\, (f\,x)\, (f\,x)$, then
   $\asize{f x}$ could be exponentially larger than $\asize{x}$. 
   Theorem~\ref{t:RS-:size} implies that no such $f$ is $\RSi$-definable, but intuitively the reason is  that 
   a $\fold_\Nat$ definition provides \emph{one} reference to the
   result of the recursive call since the $\Succ$-constructor is
   unary.  This one reference can be used multiple times, but always
   representing links to the same result, and hence not increasing the
   size.  The function that \emph{is} $\RSi$-definable is (in effect)
   $f'\Zero = \Leaf$ and $f'(\Succ x) = \Let r = f'(x) \In \Fork(r,r)$.   
   
 \item 
    \emph{Bellantoni and Cook's Raising Rule.}  It amounts to a (sound!) 
    specialization of  Whitehead and Russell's \emph{Axiom of 
    Reducibility}.  Compare the end of the first paragraph of 
    \cite[\S5]{BellantoniCook} and $\ast$12.1 of 
    \emph{Principia Mathematica}, Vol.~1, 1/e, Cambridge University Press, 
    1910,  available from 
    \texttt{http://name.umdl.umich.\allowbreak{}edu/AAT3201.0001.001}.

\item 
  Consider $\mathit{leaves}\of\Tree\to\Nat$ where 
  $\mathit{leaves}(t) =$ the number of leaves of $t$.
  $\RSi$ cannot compute this because $\asize{\mathit{leaves}(t)}$
  can   exponentially-larger than $\asize{t}$.
  In contrast, Pola can compute this as Pola
  allows some forms of change-of-parameter in recursions and,
  under Pola, $t$ is always a strict tree.
\item 
  \emph{Codata, Memoization, and Sharing.}
  Corecursions ($\unfold$s) are not memoized, but structure
  sharing is allowed in codata. 
  
\item

  \emph{Codata and poly-completeness.}
  Since type-level 1 $\RSi$ functions can have codata inputs 
  and outputs, we can translate some standard examples from
  type-2 complexity to show that $\RSi$ is missing some 
  functions over codata, where these functions' runtime
  complexity is comparable to that of $\RSi$-computable 
  functions. 
  The cure to this problem is to introduce an analogue of 
  Bellantoni's $\constr{Mod}$ function 
  \cite[Chapter 8]{bellantoni:thesis}  
  ($\constr{Mod} m\, n = m\bmod n$) or the authors' $\constr{Down}$
  function \cite[\S4]{DR:ATS:LMCS}
  ($\constr{Down} x\,y = x$, if $\asize{x}\leq\asize{y}$; 
  $\epsilon$, otherwise) both of which are 
  $(\hbox{safe}\to\hbox{normal}\to\hbox{normal})$ functions. 
  As to the motivations for such functions and their odd typing
  we refer the reader to \cite{DR:ATS:LMCS}.
  Adding such a function to $\RSi$ is \emph{not} a major change.

\end{enumerate}


The next lemma is a key property of terms with normal types.
Its proof is a simple induction on type derivations.

\begin{lemma}
\label{l:L1:sclosed} 
  If 
  $\Gamma \entails e\of\tau$ where $\tau$ is normal,
  then $\sfv(e)=\emptyset$.
\end{lemma}

\begin{lemma}[Basic Poly-Heap Bounds Arithmetic] \label{l:ph}
Suppose $\Gamma\entails e\of\sigma$, 
$\asize{e}\leq p+\asize{\sfv(e)}$, 
$\Gamma\entails e'\of\sigma'$, 
and $\asize{e'}\leq p'+\asize{\sfv(e')}$, 
where $p$ and $p'$ are  polynomials over 
$\{\,\asize{x} \suchthat \Gamma(x)$ is normal$\,\}$.
Also suppose $x\in\fv(e)$ with $\Gamma(x)=\sigma'$.
Then:
\begin{asparaenum}[(a)]
 \item \label{i:norm:sub}
    $\asize{e[x\gets e']}\leq 
     (p[\asize{x}\gets p'])+\asize{\sfv(e[x\gets e'])}$,
     if $\sigma'$ is normal.
  \item \label{i:safe:sub}
    $\asize{e[x\gets e']}\leq 
     p+p'+\asize{\sfv(e[x\gets e'])}$,
    if $\sigma'$ is safe.
  \item  \label{i:pair}
    $\asize{(e,e')} \leq p+p'+\asize{\sfv(\,(e,e')\,)}$.
\end{asparaenum}
\end{lemma}

\begin{proof}[Sketch]
\emph{Part~\eqref{i:norm:sub}:} By $\sfv(e')=\emptyset$. 
Hence, by the monotonicity of our polynomials, (a) follows.

\emph{Part~\eqref{i:safe:sub}:}
By monotonicity again 
(and some  abuse of notation):
$\asize{e[x\gets e']} \leq
p + \asize{e',(\sfv(e)-\set{x})}
\leq p+ (p'+\asize{\sfv(e'),(\sfv(e)-\set{x})})
\leq p+p'+\asize{e[x\gets e']}$.

\emph{Part~\eqref{i:pair}:}
A na\"{\i}ve upper bound on $\asize{(e,e')}$ is 
$p+p'+2\asize{\sfv(\,(e,e')\,)}$, but this double counts the
structure shared by $e$ and $e'$. 
So by eliminating the double counting, 
we have the required bound. 
\qed
\end{proof}

\emph{Poly-Heap vs.~Poly-Max Bounds.}
The analogue of parts \eqref{i:norm:sub} and \eqref{i:safe:sub} 
of Lemma \ref{l:ph} hold for poly-max bounds. 
Bounds of the form of part \eqref{i:safe:sub} are key in 
poly-boundedness arguments for forms of ``safe'' recursions. 
The analogue of Lemma \ref{l:ph}\eqref{i:pair} \emph{fails}
for poly-max bounds.  However, if one requires (\emph{\`a la} Pola)
that $e$ and $e'$ have no safe variables in common, then 
the poly-max-analogue of Lemma \ref{l:ph}\eqref{i:pair} 
does hold.
These two alternative ways of counting are at the heart of 
the $\RSi$/Pola  split.  Note that what is a stake in how one
bounds a pair is how, in general, one bounds the size of 
branching structures.



\begin{theorem}[Theorem~\ref{t:RS-:size} Restated] 
  Given an $\RSmi$ judgment $\Gamma\entails e\of \tau$  in which
  $\tau$ and each $\sigma\in\image(\Gamma)$ is a ground type, one can
  effectively find a normal polynomial $p$ with
  $\asize{e} \leq p + \asize{\sfv(e)}$.
\end{theorem}

\begin{proof}[Partial sketch]
Our first problem in exhibiting the upper bound is that $e$ may well 
contain higher-type subterms.  
Let $\tilde{e}$ be the normalized version of $e$.  Note
that $|e|\leq |\tilde{e}|$, where $|\tilde{e}|$ can be much larger
than $|e|$.  But a poly-heap bound on $|\tilde{e}|$ serves as a
bound on $|e|$.  Thus, we assume without loss of generality that
$e$ is normalized.
Since $e$ is normalized, the only place a $\lambda$-expression can 
occur in $e$ is
as the first argument of a $\folds{}$-construct, moreover, these
$\lambda$-expressions have level-1 types.
Also note that each variable occurring in $e$ must be of ground type.

The proof is a structural induction on the derivation of 
$\Gamma\entails e\of\tau$.  We consider the last rule
used in this derivation.  

All of the cases, save one, are standard, straightforward arguments---adjusting for the change from poly-max to poly-heap 
bounds.  So we omit these. The interesting case is the one for 
$\folds{\sigma}$.  We treat this case which, for simplicity 
and concreteness, we further narrow to the case for
$\folds{\Tree}$, which 
touchs on the key issues in the general $\folds{\sigma}$-case.
Recall that
$F_\Tree X = \Unit + X \times X$, \
$F_\Tree f = \id_{\Unit} + f\times f
= \lam{u}\big(\Case u \Of\,(\iota_1\, v) \Rightarrow \iota_1()$; 
$(\iota_2\, v) \Rightarrow \iota_2(f(\pi_1(v)),f(\pi_2(v)))\big)$,
and $(\folds{\Tree} g)\circ c_\Tree = g \circ F_\Tree(\folds\Tree g)$.

\emph{Some conventions:}  To cut down on clutter, 
when $y$ is of ground type $\sigma$ and $v$ is a type-$\sigma$
value  (i.e., a pointer to an internal representation of 
a type-$\sigma$ object), we shall rewrite 
$e\theta[y\mapsto v\theta']$ to $e[y\gets v]\theta$, 
\emph{provided} the value named by $v\theta'$ is a function
of $\theta$.  The substitution of the (pointer) $v$ for
the variable $y$ in $e$ is, in essence, just cutting out
one level of indirection and thus simplifies reasoning about
the value of $e\theta[y\mapsto v\theta']$.
Similarly, in ``heap'' expressions $\asize{e_1,\dots,e_k}$
we allow value terms (i.e., pointers to representations) 
among the $e_i$'s with the obvious meaning of 
$\asize{e_1,\dots,e_k}$, again we are simply cutting out a
level of indirection.  Finally, if $E$ a set of $k$-many
expressions $e_1,\dots,e_k$, then 
$\asize{E} = \asize{e_1,\dots,e_k}$.

\textsc{Case:} $\folds{\Tree}$-\emph{I}.  
Thus, $e = (\folds{\Tree}\, (\lam{z}e_0) \, e_1)$, where 
$\Gamma\entails e\of\tau$, 
$\Gamma,z\of{F_{\Tree}\tau}\entails e_0\of{\tau}$, 
$\Gamma\entails e_1\of{\Tree}$, and $\tau$ is a safe base type.
By the induction hypothesis, there are normal polynomials $p_0$
and $p_1$ that $\asize{e_0}\leq p_0 + \asize{\sfv(e_0)}$, 
and $\asize{e_1}\leq p_1$.  Fix an environment $\theta$
and suppose $e_1\theta\yields t_1\theta'$.  Recall that
$t_1$ is a pointer to the dag-representation of $e_1$'s value.
(Since $t_1$ is a data-constant, it suffices to take $\theta'=\theta$.) 
Let $t_2,\dots,t_n$ be pointers to the other $\Tree$-cons-cells
in the representation, ordered so that, for all $i$ and $j$,
if $t_i$ is an dag-ancestor of $t_j$, then $i\leq j$. 
Suppose, for $i=1,\dots,n$, $(\folds{\Tree} (\lam{x}e_0)\,t_i)\theta\yields r_i\theta_i$, where $r_i$ is a pointer
to the dag-representation of the result of the $\folds\Tree$-recursion. \textbf{N.B.}  The $t_i$'s and $r_i$'s are functions
of $\theta$.  So, as a reminder of this, 
 in our bounds calculations, we shall make explicit
the usually suppressed $\theta$.

\smallskip\noindent\emph{Claim 1:} Suppose the 
hypotheses of Lemma~\ref{l:ph} and suppose $\sigma'$ is safe.
Then 
$\asize{\set{e[x\gets e'],e_1,\dots,e_k}}\theta\leq 
  (p+\asize{\set{e',e_1,\dots,e_k}\cup\sfv(e[x\gets e'])})\theta$,
  for all $\theta$.

\smallskip\noindent\emph{Proof:}
This is just an extension of the proof of Lemma~\ref{l:ph}(b).

\smallskip\noindent\emph{Claim 2:} For each $i=1,\dots,n$:
\begin{asparaenum}[(a)]
  \item If $t_i=\underline{\Leaf} = \uc_\Tree(\uiota_1\uep)$, 
  	then $(\folds{\Tree} (\lam{x}e_0)\,t_i)\theta 
	      = e_0\theta[z\mapsto \uiota_1\uep]$.
  \item If $t_i=\underline{\Fork}\,t_j\,t_j 
            = \uc_\Tree(\uiota_2\upair{t_j}{t_k})$, 
  	then $(\folds{\Tree} (\lam{x}e_0)\,t_i)\theta 
	      = e_0\theta[z\mapsto \uiota_2\upair{r_j}{r_k}]$
	      where $j,k>i$.
  \item $\asize{\set{r_i,\dots,r_n}\cup\sfv(e)}\theta
        \leq p_0 +\asize{\set{r_{i+1},\dots,r_n}\cup\sfv(e)}\theta$.
\end{asparaenum}

\smallskip\noindent\emph{Proof:}
Part (a) is a  straightforward calculation.  

Part (b) is another straightforward calculation, taking into 
account that the implementation of $\folds\Tree$ is memoizing.

Part (c). \emph{Case:} $t_i$ is a leaf. 
Then
\begin{align*}
\lefteqn{\asize{\set{r_i,\dots,r_n}\cup\sfv(e)}\theta }\\
& \Quad1= 
\asize{\set{e_0[z\gets \uiota_1\uep] ,r_{i+1},\dots,r_n}\cup\sfv(e)}\theta
&& \hbox{(by part (a))}
\\
& \Quad1=  
p_0 + \asize{\set{\uiota_1\uep,r_{i+1},\dots,r_n}\cup\sfv(e)}\theta
&& \hbox{(by Claim 1)}
\\
& \Quad1= 
 p_0 + \asize{\set{r_{i+1},\dots,r_n}\cup\sfv(e)}\theta
 &&\hbox{(since $\asize{ \uiota_1\uep}=0$).}
\end{align*}
\emph{Case:} $t_i$ is a fork. 
Then
\begin{align*}
\lefteqn{\asize{\set{r_i,\dots,r_n}\cup\sfv(e)}\theta }
\\
& \Quad1=  
\asize{\set{e_0[z\gets \uiota_2\upair{r_j}{r_k}] ,r_{i+1},\dots,r_n}\cup\sfv(e)}\theta
&& \hbox{(by part (b))}
\\
& \Quad1= 
 p_0 + \asize{\set{\uiota_2\upair{r_j}{r_k},r_{i+1},\dots,r_n}\cup\sfv(e)}\theta
&& \hbox{(by Claim 1)}
\\
& \Quad1=  
p_0 + \asize{\set{r_{i+1},\dots,r_n}\cup\sfv(e)}\theta
&& \hbox{(since $j,k>i$)}.
\end{align*}

Thus by Claim 2(c), $\asize{(\folds{\Tree} (\lam{x}e_0)\,e_1)\theta }
= \asize{(\folds{\Tree} (\lam{x}e_0)\,t_1)\theta }
= \asize{\set{r_1}\cup\sfv(e)}\theta
\leq \asize{\set{r_1,\dots,r_n}\cup\sfv(e)}\theta
\leq (p_0\cdot n + \asize{\sfv(e)})\theta$.
Recall that $\asize{e_1}\leq p_1$.
Therefore, $p=p_0\cdot p_1$ suffices for this case.

The effectiveness part of the theorem follows from the fact
that the induction argument essentially describes a recursive 
algorithm for constructing $p$.  \qed
\end{proof}

\begin{theorem}[Theorem~\ref{t:RS:size} Restated]
  For an $\RSi$-judgment $\Gamma\entails e\of \tau$ where
  $\tau$ and each $\sigma\in\image(\Gamma)$ is a ground type, one
  can effectively find a normal second-order polynomial $p$ such
  that, if $\tau$ is a data-type, then $\asize{e} \leq p +
  \asize{\sfv(e)}$ and, if $\tau$ is a codata-type, then
  $\osize{e} \leq \lam{n}(p + \asize{\sfv(e)})$.
\end{theorem}

\begin{proof}[Partial sketch]
As in the proof of Theorem~\ref{t:RS-:size}, we may without loss of
generality assume that $e$ is normalized.  Thus 
the only place a $\lambda$-expression can occur in $e$ is 
as the first argument of a $\folds{}$- or an $\unfolds{}$-construct, 
and moreover, these $\lambda$-expressions have level-1 types.
Also note that each variable occurring in $e$ must be of ground type. 
Our proof is a structural induction on the derivation of 
$\Gamma\entails e\of\tau$.  
We consider the last rule used in this derivation.

Now, as in our sketch of the proof of Theorem~\ref{t:RS-:size}, 
here we shall present just one key case ($\unfolds\tau$-\emph{I}), 
and in fact, a specialization of that ($\unfolds\Nats$-\emph{I}).
Unlike the situation for the proof of Theorem~\ref{t:RS-:size},
the omitted cases here are less standard and a few 
involve some fine points.  However, almost all of these omitted 
cases parallel problems we dealt with our work on feasible
type-level 2 programming formalisms 
\cite{DR:ATS:LMCS,DannerRoyer:2algs}.

\textsc{Case:} $\unfolds{\Nats}$-\emph{I}.
We consider the case where $\sigma$ is a \emph{data} type.
Thus, $e = (\unfolds{\Nats}\, (\lam{z}e_0) \, e_1)$, where 
$\sigma$ is a safe ground data-type,
$\Gamma,z\of\safe\sigma\entails e_0\of\safe\Nat\times \safe\sigma$, and $\Gamma\entails e_1\of{\safe\sigma}$.  
Recall 
$F_\Nats X = \Nat \times X$,
$F_\Nats f = \id_{\Nat}\times f
= \lam{u}(\pi_1 u,f(\pi_2 u))$, and
$\co{d}_\Nats\circ (\unfolds\Nats\,g) = 
F_\Nats (\unfolds{\Nats} g) \circ g =
\lam{u}(\pi_1(g(u)),\allowbreak \unfolds\Nats\, g\,(\pi_2(g(u))))$.
Let $g_1=\pi_1\circ g$ and $g_2=\pi_2\circ g$, then
for all $n\geq 1$:
\begin{gather}\label{e:nth}
  \co{d}_\Nats^{(n)}(\unfolds\Nats\,g\;u) \;=\; 
  \left(g_1(g_2^{(n-1)}u),\;
  \unfolds\Nats\, g\;(g_2^{(n)}u)\right).
\end{gather}
Now, by the induction hypothesis, there are normal polynomials
$p_0$ and $p_1$ such that
$\asize{e_0}\leq p_0+\asize{z,\sfv(e)}$ and
$\asize{e_1}\leq p_1+\asize{\sfv(e)}$.  
By \eqref{e:nth}, to bound $\osize{e}(n)$ for $n\geq 1$, 
it suffices to bound $\asize{g_1(g_2^{n-1} e_2)}$ for 
$g_1 = \pi_1\circ (\lam{z} e_0)$
and $g_2 = \pi_2 \circ (\lam{z} e_0)$.
For $n=1$, 
$\asize{g_1(g_2^{n-1} e_1)} = \asize{e_0[z\gets e_2]}
\leq p_0 + \asize{e_1,\sfv(e)} \leq p_0+p_1+\asize{e}$.
Iterating this, we have for $n\geq1$, 
$\osize{e}(n) \leq \asize{g_1(g_2^{n-1} e_1)} \leq p_0 
+ p_1\cdot n + \asize{\sfv(e)}$.
Hence, $p=p_0+p_1\cdot n$ suffices.

In the case where $\sigma$ is a \emph{codata} type, the basic 
structure of the argument stays the same but the (second-order
polynomial) algebra becomes more involved.

%

The induction above essentially describes a recursive algorithm 
for constructing $p$.  Hence, the effectiveness part of the 
theorem follows. 
\qed
\end{proof}

\end{document}